\documentclass[11pt]{article}

\usepackage{geometry}
\usepackage{graphicx}
\usepackage[misc,geometry]{ifsym}
\usepackage{tikz}
\usetikzlibrary{arrows}
\usepackage{url}
\usepackage{verbatim}
\usepackage{amsfonts}
\usepackage{caption}
\usepackage{tabularx}
\usepackage{array}
\usepackage{booktabs}
\usepackage{multirow}
\usepackage{fancyvrb}
\usepackage{amsmath}
\usepackage{comment}
\usepackage{algorithm}
\usepackage{algpseudocode}
\usepackage{enumerate}
\usepackage{amsthm}
\usepackage{latexsym}
\usepackage{amsmath,amssymb}

\def\E{\mathbf{E}}

\def\Pr{\mbox{{\bf P}}}
\def\Prob{\Pr}
\def\whp{{w.h.p.}}

\newcommand{\ind}[1]{\mbox{{\large 1}} _{#1}}
\def\bigO{\mathcal O}

\newenvironment{proofof}[1]{{\it \proofname\ of #1.}~}{\qed}
\newcommand{\ignore}[1]{}

\newtheorem{theorem}{Theorem}
\newtheorem{lemma}{Lemma}
\newtheorem{corollary}{Corollary}

\newcounter{rot}

  \def\d{\delta} 
 
\def\f{\phi}   
\def\G{\Gamma}  

 \def\th{\theta}    \def\l{\lambda}
   
\def\r{\rho}
  
   \def\Om{\Omega}

\newcommand{\wh}[1]{\widehat{#1}}

\newcommand{\brac}[1]{\left(#1\right)}

\newcommand{\beq}[1]{\begin{equation}\label{#1}}
\newcommand{\eeq}{\end{equation}}

\def\COBRA{COBRA }
\def\Cobra{COBRA}

\def\BIPS{{BIPS} }
\def\Bips{{BIPS}}

\def\COV{\mbox{$\text{COV}$}}
\def\cov{\mbox{$\text{\textbf{cov}}$}}
\def\Hit{\mbox{$\text{Hit}$}}
\def\INF{\mbox{$\text{Infec}$}}
\def\inf{\mbox{$\text{\textbf{infec}}$}}

\begin{document}

\title{The Coalescing-Branching Random Walk on Expanders and 
the Dual Epidemic Process
}

\author{
Colin Cooper\thanks{Department of Informatics, King's College London, UK. {\tt colin.cooper@kcl.ac.uk}}
\and Tomasz Radzik\thanks{Department of Informatics, King's College London, UK. {\tt tomasz.radzik@kcl.ac.uk}}
\and Nicol\'as Rivera\thanks{Department of Informatics, King's College London, UK.  {\tt nicolas.rivera@kcl.ac.uk}}
\thanks{Research supported
by EPSRC grant EP/M005038/1,
``Randomized algorithms for computer networks''.
Nicol\'as~Rivera was supported by funding from Becas CHILE.
}}

\date{23 May 2016}

\maketitle

\begin{abstract}
Information propagation on graphs is a fundamental topic in distributed computing. 
One of the simplest models of information propagation is the push protocol in which 
at each round each agent independently pushes the current knowledge to a random neighbour. 
In this paper we study the so-called coalescing-branching random walk (\Cobra), 
in which each vertex pushes the information to $k$ randomly selected neighbours and then  
stops passing information until it receives the information again.
The aim of \COBRA is to propagate information fast but with a
limited number of transmissions per vertex per step. In this paper we study 
the cover time of the \COBRA  process defined as the minimum time until each vertex 
has received the information at least once.
Our main result says that if $G$ is an $n$-vertex $r$-regular graph whose transition matrix 
has second eigenvalue $\lambda$, then the \COBRA cover time of $G$ is $\bigO(\log n )$, 
if $1-\lambda$  is greater than a positive constant,
and $\bigO((\log n)/(1-\lambda)^3))$, if $1-\lambda \gg \sqrt{\log( n)/n}$.
These bounds are independent of $r$ and hold for $3 \le r \le n-1$.
They improve the previous bound of $O(\log^2 n)$ for expander graphs [Dutta et al., SPAA 2013].

Our main  tool in analysing the \COBRA process is
a novel duality relation between this process and a discrete epidemic process, which we
call a biased infection with persistent source (\Bips).
A fixed vertex $v$ is the source of an infection and remains permanently infected.
At each step each vertex $u$ other than $v$ selects $k$ neighbours, independently and uniformly, 
and $u$ is infected in this step if and only if at least one of the selected neighbours 
has been infected in the previous step.
We show the duality between \COBRA and \BIPS which
says that the time to infect the whole graph in the \BIPS process is of
the same  order as the cover time of the \COBRA process.

\end{abstract}

\section{Introduction}

Dutta et al.~\cite{COBRA,Dutta:2015:CRW:2821462.2817830} studied the following
{\em coalescing-branching} random walk process for propagating information
through a connected $n$-vertex graph.
At the start of a round each vertex containing information ``pushes'' this information to
$k$ randomly selected neighbours, then it stops passing the information until it receives the information again.
At the end of a round if a vertex receives information from two or more vertices,
then the information coalesces into one. Thus it does not help if a vertex receives
the same information from more than one neighbour. The continuous act of coalescing and branching
gives the name COBRA to this process.

The aim of the \COBRA process is to rapidly propagate information to all vertices
but to limit the number of transmissions per vertex per step
and without requiring that vertices store information for longer than one round.
In the special case that $k=1$, the \COBRA process is a simple random walk,
which achieves a low transmission rate but does not satisfy the fast propagation condition.

The main quantity of interest in
information propagation processes
is the time taken to inform (or  visit) all vertices. By analogy with a random walk,
this is referred to as the {\em cover time}.
The w.h.p.\footnote{``With high probability,'' which means in this paper probability at most $n^{-c}$,
for some positive constant $c$.} cover time results
for the \COBRA process obtained in \cite{COBRA,Dutta:2015:CRW:2821462.2817830} 
for the case $k=2$ include the following.
(i)
For the complete graph $K_n$ all vertices are visited in $\bigO(\log n)$ rounds.
(ii)
For regular constant degree expanders, the cover time is $\bigO(\log^2 n)$.
(iii)
For the $d$-dimensional grid, the cover time is $\tilde{O}(n^{1/d})$.
By comparison with the complete graph, it might seem that the cover time of any
$r$-regular expander by the \COBRA process (with $k=2$) should be $\bigO(\log n)$ for any  degree $r$ between $3$ and $n-1$. The proof of this is the main content of this paper (see Theorem \ref{Th1} below).
This is the best possible asymptotic bound
since
the number of visited vertices at most doubles in each round.

The  \COBRA process imitates a type of epidemic process but with an upper bound $k$ on the number of contacts.
Indeed, the \COBRA process turns out to be a discrete version of the
{\em contact process}, which is a continuous  model with exponential
waiting times, in which (typically) a particle at vertex $v$ infects
each neighbour with rate 
$\mu$, and becomes extinct with rate 1.
One difference between the processes is that a contact process can die out, whereas the \COBRA one does not.
The contact process was introduced by Harris in 1974 \cite{Har},
and has been extensively  studied on infinite lattices and  trees.
A major topic of study is, given the initial spread of infection, to
determine the values of $\mu$ for which the process is transient, recurrent, and stationary.
See for example, Madras and Schinazi \cite{MadS} for a concise summary, and  also Liggett \cite{Lig}; Pemantle \cite{Pem} gives a  more detailed analysis on trees, Liggett for finite trees \cite{Lig2}. The work of Bezuidenhout and  Grimmett \cite{BZG} was a breakthrough paper for lattices.

We proceed with the formal definition of the COBRA process and the statement of our main results.

\textbf{Coalescing Branching Random Walk (COBRA):} Consider a graph $G = (V,E)$ and an integer $k \geq 1$. Let $C \subseteq V$ and consider the set process $(C_t)_{t \geq 0}$ with $C_0 = C$
and $C_{t+1}$ is defined as follows. Let $C_t$ be the vertices chosen at round $t$ (not necessarily for the first time).
Each vertex $v \in C_t$ independently chooses $k$ neighbours uniformly at random with replacement and all the chosen vertices belong to $C_{t+1}$.

For $C_0 = \{ u \}$, let $\cov(u) = \min\{ T: \bigcup_{t=1}^{T} C_t = V\}$
be the number of steps needed for the \COBRA process
to visit all vertices of the graph $G$ starting from vertex $u$; and let $\COV(u) = \E(\cov(u))$.
By analogy with the cover time of  a random walk, which measures the worst case
starting vertex, we let $\COV(G)=\max_{u \in V} \COV(u)$ be the cover time of the \COBRA process.

Let $G$ be a connected $n$-vertex $r$-regular graph with adjacency matrix $A(G)$
and random-walk transition matrix $P=A(G)/r$. Let $\l_1, \l_2,...,\l_n$
be the eigenvalues of the transition matrix ordered in a non-increasing sequence.
Thus $\l_1=1$, $\l_n \ge -1$. Let
$\l=\l_{\max}=\max_{i=2,...,n} |\l_i|$ be the second largest eigenvalue (in absolute value).
If the graph $G$ is not bipartite, then $\l <1$.  In which case we have the following theorem.

\begin{theorem}\label{Th1}
Let $G$ be a connected regular $n$-vertex graph with 
$1-\lambda\gg \sqrt{(\log n)/n}$. \footnote{%
$(1-\lambda) \ge  C\sqrt{(\log n)/n}$ for some suitably large constant $C$.}
Let
\begin{eqnarray*}
T= \frac{\log(n)}{(1-\lambda)^3}.
\end{eqnarray*}
Let $\COV(G)=\max_{u \in V} \COV(u)$ be the cover time of $G$ by a \COBRA process with branching factor $k=2$. Then $COV(G)=\bigO(T)$ and for all $u \in V$, w.h.p. $\cov(u)=\bigO(T)$.
\end{theorem}

The COBRA process is a type of a multiple random walk processes, so it is tempting to
try to analyse COBRA using techniques developed for such processes. Previous work
on multiple random walks includes \cite{%
DBLP:journals/cpc/AlonAKKLT11,%
DBLP:conf/stoc/BroderKRU89,%
DBLP:journals/siamdm/CooperFR09,%
DBLP:conf/icalp/ElsasserS09}, where cover times were analysed for various classes of graphs.
The analysis of the COBRA process given in Dutta et al.~\cite{COBRA,Dutta:2015:CRW:2821462.2817830} 
uses  a number of tools
from multiple random walks, but applicability of those tools turns out to be limited because
the random walks in COBRA are highly dependent.
In order to prove Theorem \ref{Th1}, we introduce a related epidemic process BIPS, which
is a dual of \COBRA under time reversal, and work on this new process instead of the original one.
The formal duality between \BIPS and \COBRA used in the proof of Theorem \ref{Th1}
is established  in Theorem \ref{CoBips}.

\textbf{Biased Infection with Persistent Source (BIPS):} Consider a graph $G = (V,E)$ and an integer $k \geq 1$. Consider a vertex $v$ which is the source of an infection. We consider the process $A_t$ defined by $A_0 = \{v\}$. Given $A_t$ each vertex $u \in V$, other than $v$,  independently and uniformly with replacement selects $k$ neighbours and becomes a member of $A_{t+1}$ if and only if at least one of the selected neighbours is in $A_t$. Additionally, $v \in A_t$ for all $t \geq 0$.
We call $A_t$ the infected set at time $t$. Observe the source $v$ is always infected.
Finally, if $A_0 = \{v\}$ then it is clear that $v$ is the source of the infection process.

The \BIPS  process is a discrete epidemic process of the SIS (Susceptible-Infected-Susceptible) type,
in which vertices (other than the source $v$) refresh their infected state at each step
by contacting  $k$ randomly chosen neighbours. The presence of a persistent (or corrupted) source
means that w.h.p. all vertices of the underlying graph eventually become infected.
The \BIPS process is of independent interest since
in the context of epidemics, certain viruses exhibit the property that a particular host can become persistently infected. For example,
in animals the BVDV (Bovine viral diarrhea virus) is of this type
and a model that mimics the spread of BVDV
was described in~\cite{IMBG}.
The model is able to simulate the spread of infection when a persistently infected animal is introduced
into an infection-free herd.

If we define $\inf(v)$ as the first time when all vertices are infected when the source is $v$,
 $\INF(v) = \E(\inf(v))$, and $\INF(G)= \max_{u \in V} \INF(v)$, then we have the analogue theorem of Theorem \ref{Th1}.

\begin{theorem}\label{Th2}
Let $G$ be a connected $n$-vertex $r$-regular graph with $1-\lambda \gg \sqrt{(\log n)/n}$. Then  for every
$v \in V$, the infection time, $\inf(v)$, by a \BIPS process with $k=2$ satisfies $\inf(v) = \bigO(\log(n)/(1-\lambda)^3)$ in expectation and with probability at least $1-\bigO(1/n^3)$. Moreover $\INF(G)= \bigO(\log(n)/(1-\lambda)^3)$.
\end{theorem}

Although Theorem \ref{Th1} is proved for a \COBRA process with
branching factor $k=2$, it seems natural to ask if, for expanders, a cover time of $\bigO(\log n)$ can be obtained with less branching. Clearly $k=1$ is not enough; the cover time of any $n$ vertex graph by a random walk is $\Om(n \log n)$. Suppose that at the start of each step, each particle divides in two with probability $\r$.
This gives an expected branching factor of $1+\r$. The following result shows that any constant $\r>0$ will do. The proof of Theorem \ref{Th3} follows from the
proof of Theorem \ref{Th1}, by using Corollary \ref{Corol} of Lemma \ref{Eval} in Section \ref{EA}.

\begin{theorem}\label{Th3}
Let $G$ be a connected $n$-vertex $r$-regular graph with $\lambda<1$ constant and let $v \in V(G)$. Provided  $\r > 0$ constant, for all $v \in V$, the cover time of $G$ by the \COBRA process with branching factor $1+\r$ and $C_0 = \{v\}$ is $\cov(v)=\bigO(\log(n))$ rounds in expectation and with high probability.
\end{theorem}

\subsection*{Overview of the proof of Theorem \ref{Th1}}

For a connected graph $G=(V,E)$, denote $\Hit_C(v) = \min\{t: v \in C_t, C_0=C\}$ the hitting time of vertex $v$ for a \COBRA process starting from $C$.
In particular, $\Hit_u(v)$ is the hitting time of $v$ starting from $u$. Suppose there exists $T >0 $ such that for any fixed pair $(u,v)$ it holds that $\Hit_u(v) >T$ with probability at most $\bigO(1/n^2)$. By the union bound we get  $\cov(u) = \max_{v \in V} \Hit_u(v)>T$ with probability at most $\bigO(1/n)$. Thus $\cov(u) < T$ with high probability. This results also holds in expectation. To see this, observe that by restarting the process after $T$ steps with any of the existing particles, we have that
\begin{eqnarray}
\COV(u) & \le & T + \bigO(1/n) 2T + \cdots +\: \bigO(1/n^j) (j+1) T + \cdots \; = \; \bigO(T). \label{eqn:COVExpectedValue}
\end{eqnarray}
Thus obtaining that $\COV(u) = \bigO(T)$. In our case $T = \bigO(\log(n)/(1-\lambda)^3)$.

We give an overview of the Proof of Theorem \ref{Th1}.
\begin{enumerate}
\item In Section \ref{Models} we prove Theorem \ref{CoBips}
which relates hitting times in the \COBRA process to the membership of elements of the infected set at a given step of the \BIPS process as follows
\begin{equation}\label{bkp09e23}
\Prob(\Hit_u(v)> t) = \Prob(u \not \in A_t | A_0 = \{ v\} ).
\end{equation}
The probability on the left hand side is for \COBRA process staring from $u$, and on the right hand side for a \BIPS process with persistent source $v$, and where  $A_t$ the infected set at step~$t$.

\item From Theorem~\ref{Th2} we obtain
that for $T$ as above,
$\Prob(u \not \in A_T | A_0 = v) \le \Prob(A_T \neq V | A_0 = v) = \bigO(1/n^3)$.
This together with the above duality implies Theorem~\ref{Th1}.
Sections~\ref{EA}--\ref{Fast-2} are devoted to the proof of Theorem~\ref{Th2}.

\item In Section \ref{EA} we prove a lower bound on $\E(|A_{t+1}| \mid |A_t|)$, the expected
size of the infected set at round $t+1$ of a \BIPS process given the size at step $t$.

\item
The proof of Theorem~\ref{Th2} is split into two parts.
In Section~\ref{SmallPhase} we show that
in $\bigO(T)$ rounds the infected set $A_t$ increases its size from $1$ to $\Omega(T)$
(we refer here only to the case when $1-\l$ is constant).
In Section~\ref{Fast-2} we prove that  in $\bigO(T)$ extra rounds the whole graph is infected.
Both parts of the proof of Theorem~\ref{Th2} use the bound on $\E(|A_{t+1}| \mid |A_t|)$.

\end{enumerate}

\textbf{Notation.}
For a vertex $x\in V$, $N(x)$ denotes the set of neighbours of $x$
and $d(x) = |N(x)|$. More generally,
for $A \subseteq V$, $d_A(x) = |N(x) \cap A|$.
We also define $d(A) = \sum_{x \in A} d(x)$. If there is no ambiguity, we write $A$ in place of $|A|$.
Let $\G(A)=\{x \in V: \exists u \in A, \{x,u\} \in E(G)\}$
be the inclusive neighbourhood of $A$.
By replacing each undirected edge $\{x,u\}$ with directed edges $(x,u)$ and $(u,x)$
and counting the edges with the second vertex in $A$, it follows that
 $\sum_{u \in \G(A)}d_A(u)=r |A|$.

\section{Duality between COBRA and BIPS processes}\label{Models}

The intuition for the next theorem can be seen as follows. Replace each edge in the graph 
by a pair of directed edges. Let $v$ be a distinguished vertex. Delete all out edges of $v$ 
and replace them with a loop. Let $k=1$. Any random walk which arrives at $v$ at or before 
step $t$ remains at $v$. A walk $W$ starting at $u$ which arrives at $v$ at step $s<t$
corresponds to a BIPS process which remains at $v$ for $t-s$ steps and then follows 
the edges of $W$ back to $u$ in $s$ steps.
This correspondence generalizes to the \COBRA and \BIPS processes with the parameter $k \ge 2$.

Recall that $(C_t)_{t \ge 0}$ and $(A_t)_{t \ge 0}$ denote the
\COBRA and \BIPS processes, respectively.
To avoid confusion, we use the notation $\Prob( \cdot)$ for probabilities in the BIPS process, and
$\wh \Prob(\cdot)$ in the COBRA process.
Our main Theorem~\ref{Th1} follows from the duality between these two processes expressed in~\eqref{bkp09e23}.
To prove~\eqref{bkp09e23}, we generalize this relation in the following theorem to a form which is convenient
for an inductive proof.
To simplify notation, we will write ``$A_0 = v$'' for the frequently appearing condition ``$A_0 = \{ v \}$.''


\begin{theorem}\label{CoBips}
Let $G$ be a connected regular graph and consider the \COBRA and \BIPS processes on $G$ with parameter $k \ge 1$.
For each $v \in V$, $C \subseteq V$ and $t \geq 0$ we have
$$\wh \Prob(\Hit_C(v)> t| C_0 = C) \;\; = \; \; \Prob(C \cap  A_t = \emptyset| A_0 = v).$$
\end{theorem}
\begin{proof}
Observe that the claim is trivial if $v \in C$, since both probabilities are 0. We assume that $v \not \in C$ and proceed by induction on $t$.
For $t = 0$ the claim is true because both probabilities are 1.
Assume the claim is true for a fixed $t \ge 0$, we will prove it for $t+1$.

Consider the BIPS process at step $t+1$.
Denote by $B_x$ the random $k$-set of neighbours chosen by vertex $x$.
Note that $B_v=\{v\}$ always.
Define $X(C)= 
\bigcup_{x \in C} B_x$.
It is an assumption of the model that, at step $t+1$,  for any fixed set $B$,
the event $X(C)=B$
is independent of $A_t$ and thus of the event $B \cap A_t=\emptyset$. Thus,
\begin{eqnarray*}
\Prob(B \cap A_t = \emptyset \text{ and } X(C)=B | A_0=v)
 & =& \Prob(B \cap A_t  = \emptyset  | A_0=v)\Prob( X(C)=B | A_0=v).
\end{eqnarray*}
Let $N^X(C)$ be the set of possible $B=X(C)$ generated by $C$ in one step:
$N^X(C) = \{ B\subseteq V:\: \Prob(X(C) = B) > 0\}$.
The events $X(C)=B$ are mutually exclusive and exhaustive, so
\begin{eqnarray*}
\Prob(C \cap  A_{t+1} = \emptyset| A_0 = v)
 & = & \!\!\!\! \sum_{B \in N^X(C)} \Prob( B \cap A_{t} = \emptyset| A_0 = v) \Prob (X(C)=B| A_0=v).
\end{eqnarray*}
For any set $B$, the induction hypothesis gives
\begin{eqnarray}
\Prob \left( B \cap A_{t} = \emptyset| A_0 = v \right) &=& \wh \Prob(\Hit(v) > t| C_0 =B)\nonumber.
\end{eqnarray}
Let $B=Y(C)$ be the random set chosen (to push to) by $C$ in the COBRA process, and let
$N^Y(C)$ be the set of possible $B$ resulting from this.
If $v \not \in C$, then for any $u \in C$
\[
\Prob(X(u)=B_u)= \wh \Prob(Y(u)=B_u).
\]
The events $B_u, u \in C$ are independent, so
\[
\Prob (X(C)=B| A_0=v)= \wh \Prob (Y(C)=B).
\]
Moreover $N^X(C)=N^Y(C)$. Thus for $v \not \in C$,
\begin{eqnarray*}
\Prob(C \cap  A_{t+1} = \emptyset| A_0 = v)
   &=& \!\!\!\!
\sum_{B \in N^X(C)} \Prob( B \cap A_{t} = \emptyset| A_0 = v) \Prob (X(C)=B| A_0=v)\\
  &=& \!\!\!\!
\sum_{B \in N^Y(C)} \wh \Prob(\Hit(v) > t| C_0 =B) \wh \Prob (Y(C)=B)\\
&=& 
 \wh \Prob(\Hit(v) > t+1| C_0 =C).
 \end{eqnarray*}
\end{proof}

\begin{sloppy}

\begin{proofof}{Theorem \ref{Th1}} 
From Theorem~\ref{CoBips} for $C = \{ u \}$
and the union bound we get that for any two vertices $u,v~\in~V,$ and any $T \ge 0$,
\begin{eqnarray*}
\wh \Prob(\Hit_u(v)> T) &=& \Prob(u \not\in A_T| A_0 = v) \\
  & \leq & \Prob(A_T \neq V| A_0 = v) \\
  &  = & \Prob(\inf(v) > T).
\end{eqnarray*}
For $T = \bigO((\log n)/(1 - \l)^3)$, 
Theorem~\ref{Th2} tells us that  $\Prob(\inf(u) > T) = \bigO(1/n^3)$, so $\wh\Prob(\Hit_u(v) > T) = \bigO(1/n^3)$. Therefore
$$\Prob(\cov(u)> T) \;\; \leq \;\; \sum_{v \in V} \Prob(\Hit_u(v) > T) \; = \; \bigO\left(\frac{1}{n^2}\right).$$
Concluding that $\cov(u) \leq T$ with probability at least $1- O(1/n^2)$. 
From Equation~\eqref{eqn:COVExpectedValue}~we get the result in expectation.
\end{proofof}

\end{sloppy}

\ignore{
We obtain the following corollary. Let $\tau_v = \min\{t \geq 0 : A_t = V\}$, i.e. $\tau_I$ is the first time such that the whole graph is infected $A_t$. 

\begin{corollary}\label{cor:union1}
If $\tau_v >T$ with probability at most $o(1/n^2)$, then $\MHit(u) > T$ with probability at most $o(1/n)$.
\end{corollary}

\[
\Prob(\Hit(v) > T|C_0 = \{u\})=
 \Prob(u \not \in A_T |A_0 = \{v\}) \leq o(1/n^2).
\]
From there
\[
\Prob(\MHit(u) > T) \leq \sum_{v \in V} \Prob(\Hit(v)>T|C_0 = u) = o(1/n)
\]

 \begin{proof}
Take $C = 30((2-\lambda^2)/(1-\lambda^2)^2)r$ in Lemma~\ref{lemma:fastphase1}. Then with probability at least $1-\bigO(1/n^2)$ we reach an infection of size at least $C' = 26(2-\lambda^2)/(1-\lambda^2)^2)\log(n)$ which is enough to apply Lemma~\ref{lemma:fastphase2}, thus finishing the infection in $\bigO(\log(n))$ rounds with probability at least $1-\bigO(1/n^2)$. Therefore, the infection process finishes in $\bigO(\log(n))$ rounds with probability at least $1-\bigO(1/n^2)$. From Corollary~\ref{cor:union1} we finish the proof.
\end{proof}

  }

\begin{sloppy}

\section{
Expected growth of the \BIPS process}\label{EA}

\end{sloppy}

The following lemma, which gives a lower bound on the expected increase of infection in one step of the \BIPS process, 
is the basis for our analysis of this process.

\begin{lemma}\label{Eval}
Let $G$ be a connected $r$-regular graph on $n$ vertices,
with $\lambda < 1 $ where $\lambda$ is the absolute second eigenvalue of the random-walk  transition matrix.
Let $A_t$ be the size of the infected set after step $t$ of the \BIPS process with $k=2$, then
\begin{eqnarray}
\E(|A_{t+1}|\mid A_t = A)  \nonumber 
  &\geq&
|A| (1 +(1-\lambda^2)(1-|A|/{n})).
\label{eqn:boundSizeAt|At-1}
\end{eqnarray}
\end{lemma}
\begin{proof}
Note that $\sum_{u \in \G(A)}d_A(u)=r |A|$. Thus
\begin{eqnarray}
\E (|A_{t+1}| \mid A_t=A)  \label{Lab1} 
  &=& 1 + \sum_{u \in \G(A)\setminus\{v\}} \left( 1-(1-d_A(u)/r)^2 \right)  \\
&\ge& \sum_{u \in \G(A)}\brac{ \frac{2}{r} d_A(u)-\frac{1}{r^2}d_A^2(u)} \\
& = & 2|A|-\frac{1}{r^2}\sum_{u \in \G(A)}d_A^2(u).\label{Lab2}
\end{eqnarray}
\ignore{
We first consider the case that  $|A| \le r$. Thus $d_A(u) \le |A|$ and
$\sum_{u \in \G(A)} d_A^2(u) \le |A| \sum_{u \in \G(A)}d_A(u) = r|A|^2$, giving
\[
\E (|A_{t+1}| \mid A_t=A) \geq |A|(2-|A|/r).
\]
}
For the  inequality in \eqref{eqn:boundSizeAt|At-1} we argue as follows.
Let $P=P(G)$ be the transition matrix of a simple random walk on $G$.
Let $P(x,A)=\sum_{y \in A}P(x,y)=d_A(x)/r$. From \eqref{Lab1}-\eqref{Lab2}
we have
\begin{equation}
\E(|A_{t+1}|\:|A_t = A)
\ge 2A-\sum_{x \in V}P(x,A)^2. \label{eqn:boundSize1}
\end{equation}
Observe that $\sum_{x \in V}P(x,A)^2 = \langle P\ind A,P\ind A \rangle = \|P \ind A \|^2$, where $1_A=(1_{\{ x \in A\}}: x\in V)$ and $P\ind A=(P(x,A): x \in V)$.
As $P$ is symmetric, it has an orthonormal basis of right eigenvectors $f_1,...,f_n$, i.e. $\|f_i\|=1$,
$\langle f_i,f_j \rangle=0$ for $i \ne j$.
For any vector $g$, $g=\sum_{i=1}^n \langle f,f_i \rangle f_i$
and $\| g \|^2 = \sum_{i=1}^{n} \langle g, f_i \rangle^2$.
Here $f_1=(1/\sqrt n)$ is the unique eigenvector  with eigenvalue 1, and
$\langle 1_A,f_1 \rangle= A/\sqrt n$.
Thus
\begin{eqnarray}
\| P\ind A \|^2 &=& \| P\sum_{i=1}^n \langle \ind A, f_i \rangle f_i \|^2
   \;  = \; \| \sum_{i=1}^n \langle \ind A, f_i \rangle Pf_i \|^2 \nonumber \\
 & = & \| \sum_{i=1}^n \langle \ind A, f_i \rangle \lambda_i f_i \|^2
 \; =  \; \sum_{i=1}^n \langle \ind A, f_i \rangle^2 \lambda_i^2 \|f_i \|^2\nonumber\\
 &\leq& (1-\lambda^2)\langle \ind A, f_1 \rangle^2 +\lambda^2\sum_{i=1}^n \langle \ind A, f_i \rangle^2               
    \nonumber\\
 & = & (1-\lambda^2) \frac{|A|^2}{n} + \lambda^2 \|\ind A\|^2  \nonumber\\
 & = & (1-\lambda^2) \frac{|A|^2}{n} + \lambda^2 |A|.\label{bnnmwe56a}
\end{eqnarray}
Thus~\eqref{eqn:boundSize1} and~\eqref{bnnmwe56a} imply
\begin{eqnarray*}
\E(|A_{t+1}|\: |A_t = A) & \geq &  2|A|-\lambda^2 |A| -(1-\lambda^2)\frac{|A|^2}{n},
\end{eqnarray*}
which is equivalent to~\eqref{eqn:boundSizeAt|At-1}.
\end{proof}

The following corollary is easily obtained from the proof of Lemma~\ref{Eval}.
In a \BIPS process with
$k=1+\r$,  each vertex contacts one randomly chosen neighbour,
and with probability $\r>0$ randomly chooses a second neighbour, with replacement. \begin{corollary}\label{Corol}
Let $A_t$ be the size of the infected set after step $t$ of the \BIPS process with expected branching factor
$k=1+\r$, then
\begin{eqnarray*}
\E(|A_{t+1}\mid |A_t = A) &\geq&
|A| (1+\r(1-\lambda^2 )(1-{|A|}/{n})).
\end{eqnarray*}
\end{corollary}
\begin{proof}
\ignore{
For the first case, we have
\[
\E (|A_{t+1}| \mid A_t=A) = \sum_{u \in \G(A)} 1-[(1-\r)(1-d_A(u)/r)
+\r(1-d_A(u)/r)^2].
\]
}
The probability that
$x$ chooses at least one vertex in the infected set $A$, is
\begin{eqnarray*}
1-(1-P(x,A))(1-\r P(x,A)) 
 & =& (1+\r)P(x,A)-\r P(x,A)^2.
\end{eqnarray*}
The rest of the proof is the same.
\end{proof}

\section{\BIPS process for small sets}\label{SmallPhase}
Consider a \BIPS process $A_t$ with source $v$ on a graph $G$. Thus at $t=0$, $A_0=\{v\}$.
\begin{lemma}\label{lemma:fastphase1}
Let $G$ be a connected $r$-regular graph on $n$ vertices,
with $\lambda=\lambda(n) < 1$, and let $m \leq n/2$.
Then with probability at least $1-\bigO(1/n^C)$, we have that $A_t > m$
for some $t \le T$, where
$T = 13m/(1-\lambda) + 24C\log(n)/(1-\lambda)^2$.
\end{lemma}

\begin{sloppy}

\begin{proof}
Let  $A_0$ be the initial infected set and $A_t$ be the infected set at the end of round $t$.
For convenience  we denote by $A_t$ the size of the set $A_t$, instead of $|A_t|$.

Denote the event $E_t = \{A_0< m+1, \ldots, A_t < m+1\}$. We need to find an upper bound for $\Prob(E_t)$ for any $t$. Observe $\Prob(E_0)=1$, so we concentrate on $t \geq 1$. Let $\f>0$ (to be chosen later), then
\begin{eqnarray}
\Prob(E_t) &=& \Prob(E_{t-1}, A_t < m+1) \nonumber \\
 & = &  \Prob(E_{t-1}, A_t - A_0< m)\nonumber \nonumber \\
& = & \Prob(E_{t-1}, e^{-\phi(A_t-A_0)} > e^{-\phi m}) \nonumber \\
& = & \Prob( e^{-\phi(A_t-A_0)} \ind{\{E_{t-1}\}} > e^{-\phi m})\nonumber\\
&\leq& e^{\phi m} \; \E(e^{-\phi(A_t-A_0)} \ind{\{E_{t-1}\}})
\label{eqn:boundprobability}.
\end{eqnarray}
Because $\ind{\{E_{t-1}\}} = \ind{\{E_{t-2}\}} \ind{\{A_{t-1} < m+1\}}$,
and assuming $\ind{\{E_{-1}\}}\equiv 1$, we can write
\begin{eqnarray}
G_t(\f)  & \equiv  & \E(e^{-\phi(A_t-A_0)} \ind{\{E_{t-1}\}}) 
 \; = \; \E(e^{-\phi(A_t-A_{t-1})}e^{-\phi(A_{t-1}-A_0)} 
\ind{\{E_{t-2}\}} \ind{\{A_{t-1} < m+1\}})\label{rewrite}
\end{eqnarray}
Observe that $G_0(\f) \equiv 1$.
Denote the sigma algebra $\mathcal F_t = \sigma(A_0, \ldots, A_t)$. By taking expectation conditional  on $\mathcal F_{t-1}$ we rewrite \eqref{rewrite} as
\begin{eqnarray}
G_t(\f)
&=&\E\,(\: \E\,(\: e^{-\phi(A_{t-1}-A_0)} \ind{\{E_{t-2}\}} \ind{\{A_{t-1} < m+1\}} 
        \, \times \, e^{-\phi(A_t-A_{t-1})} | \mathcal F_{t-1}\,)\,) \nonumber\\
&=& \E\,(\: e^{-\phi(A_{t-1}-A_0)} \ind{\{E_{t-2}\}} \ind{\{A_{t-1} < m+1\}} 
     \, \times \, \E(e^{-\phi(A_t-A_{t-1})} | \mathcal F_{t-1}\,)\,).\label{eqn:partialG_t1}
\end{eqnarray}
We derive an upper bound on $\E(e^{-\phi(A_t-A_{t-1})} | \mathcal F_{t-1})$. Observe that
\begin{equation}
E(e^{-\phi(A_t-A_{t-1})} | \mathcal F_{t-1}) \;\; = \;\; e^{\phi A_{t-1}} \;
\E(e^{-\phi A_t}| \mathcal F_{t-1}).\label{Eefa}
\end{equation}
For  $x\in V$ the events $\{x \in A_t | \mathcal F_{t-1}\}$
are independent, and since $A_t$ is a Markov chain, they depend only on $A_{t-1}$.
Thus
\begin{eqnarray}
\E(e^{-\phi A_t}| \mathcal F_{t-1})
&=& \prod_{x \in V} \E(e^{-\phi \mbox{{\small 1}}_{\{x \in A_t\}}}|\mathcal F_{t-1}) \nonumber\\
& = &  \prod_{x \in V} \brac{e^{-\phi}\Prob(x \in A_t	|A_{t-1}) + 1-\Prob(x \in A_t|A_{t-1})}\nonumber\\
&=&\prod_{x \in V} \brac{1-(1-e^{-\phi}) \Prob(x \in A_t	|A_{t-1})}\nonumber\\
& \leq &  \prod_{x \in V} \exp\{-(1-e^{-\phi})\Prob(x \in A_t|A_{t-1})\}\nonumber\\
&=&  \exp\{-(1-e^{-\phi})\sum_{x \in V}\Prob(x \in A_t|A_{t-1})\}\nonumber\\
& = & \exp\{-(1-e^{-\phi})\E(A_{t}|A_{t-1})\} \label{eqn:firstbound}.
\end{eqnarray}
Substitute \eqref{eqn:firstbound} into \eqref{Eefa} and~\eqref{eqn:partialG_t1}~to get
\begin{eqnarray}
G_t(\f)
 & \leq & \E \left[\: e^{-\phi(A_{t-1}-A_0)} \ind{\{E_{t-2}\}} 
\, \times \: \ind{\{A_{t-1} < m+1\}} e^{\phi A_{t-1}} \exp\{-(1-e^{-\phi})
\E(A_{t}|A_{t-1})\} \right]. 
\label{eqn:partialG_t2}
\end{eqnarray}
Define 
\[ \Psi(A) \; = \;
   \ind{\{A < m+1\}} e^{\phi A} \exp\{-(1-e^{-\phi}) \E(A_{t}|A_{t-1}=A)\}.
\]
Remember that, due to the source, we have $A_t>0$ for all $t \geq 0$, thus we consider $A$ with size at least $1$ in $\Psi(A)$. Denote by $\delta =\delta(A) =  ({\E(A_t|A_{t-1} = A)})/{A}$. Then
\begin{eqnarray}
\Psi(A) &=&
\ind{\{1 \le A \le m\}}
\exp\brac{ - A ((1-e^{-\phi})\d-\f)}.\label{Nice1}
\end{eqnarray}
Our next step is to find an upper bound of $\Psi(A)$ independent of $A$.
For $k\in\{1,\ldots, n\}$ define $\delta_\lambda(k) = 1+(1-\lambda^2)(1-k/n)$.
We now choose $\phi = \log(1+x)$ where $x = \frac{1-\lambda}{2} $.
Since $|A| \leq n/2$
from \eqref{eqn:boundSizeAt|At-1} of Lemma \ref{Eval}
we get
\begin{eqnarray*}
\delta(A) & \ge & \delta_\lambda(|A|)
  \;\; =\;\; 1 +(1-\lambda^2)(1-{|A|}/{{n}}) \\
  & \geq & 1+(1-\lambda)(1+\lambda)/2 \\
   & \geq & 1+\frac{1-\lambda}{2}
   \;\; = \;\; e^{\phi}.
\end{eqnarray*}
Using this with \eqref{Nice1}, we get that
\begin{eqnarray}
\Psi(A) &=&
\ind{\{1 \le A \le m\}}
\exp\brac{ - A ((1-e^{-\phi})\d-\f)} \nonumber\\
&\leq& \ind{\{1 \le A \le m\}}
\exp\brac{ - A (e^\phi-1-\f)}.
\end{eqnarray}
Observe that $f(y) = e^y-1-y > 0$, thus we take $A = 1$ in the above quantity to get
\begin{eqnarray}
\Psi(A) \; \leq \; e^{1+\phi-e^\phi} \; = \; e^{\log(1+x)-x}.
\end{eqnarray}

From Inequality~\eqref{eqn:partialG_t2} and the fact that $\Psi(A)\leq e^{\log(1+x)-x}$, we get
\begin{eqnarray}
G_t(\phi)
 &\leq& \E \left[
e^{-\phi(A_{t-1}-A_0)} \ind{\{E_{t-2}\}} \, \times \;
 \ind{\{A_{t-1} < m+1\}} e^{\phi A_{t-1}} \exp\{-(1-e^{-\phi})
\E(A_{t}|A_{t-1})\} \right]\nonumber\\
&=& \E(e^{-\phi(A_{t-1}-A_0)} \ind{\{E_{t-2}\}} \Psi(A_{t-1}))\nonumber \\
& \leq& \E(e^{-\phi(A_{t-1}-A_0)} \ind{\{E_{t-2}\}})\exp(\log(1+x)-x)\nonumber\\
&=& G_{t-1}(\phi)\exp(\log(1+x)-x).\nonumber
\end{eqnarray}
Using $G_0(\f) \equiv 1$ and induction we get 
\[ G_t(\phi) \leq e^{t(\log(1+x)-x)}.
\] 
Putting $G_t(\phi)$ into Inequality~(\ref{eqn:boundprobability}), we obtain
\begin{eqnarray}
\Prob(E_t) \;\; \leq \;\; G_t(\phi)e^{\phi m } \;\; \leq \;\; e^{m \phi+t(\log(1+x)-x)}.
\end{eqnarray}
We need to estimate the exponent. Recall that $\phi = (1+x)$ with $x = \frac{1-\lambda}{2} \leq 1/2$. For $x<1$ the terms of $\log(1+x)$ are monotone decreasing in absolute value, so that $\log(1+x) \le x-x^2/2+x^3/3$. From this, and $x \le 1/2$, the exponent can be bounded as follows,
\begin{eqnarray*}
(t+m)\log(1+x)-tx 
 &\leq& -tx + (t+m)\left(x-\frac{x^2}{2}+\frac{x^3}{3}\right) \\
 &=& x \brac{m \brac{1-\frac{x}{2}+\frac{x^2}{3}} - \frac {tx}{2}\brac{1-\frac{2x}{3}}}\\
 &\le& x \brac{ \frac{13 m}{12}- \frac{t x}{6}}.
\end{eqnarray*}
Put $x=(1-\lambda)/2$ and choose 
\[ t=13 m/(1-\l)+ (24C \log (n))/(1-\l)^2
\] 
to obtain
the result.
\end{proof}

\end{sloppy}


\section{\BIPS process for large sets}\label{Fast-2}

In this section we analyze the growth of $A_t$ in the \BIPS process 
from $A_{t'}= \Theta(\log n/(1-\lambda)^2)$ up to $A_{t''}=n$.
We start by applying Lemma~\ref{lemma:fastphase1}
with $m=K \log(n)/(1-\lambda)^2$, for some (large) constant $K$, to have \whp\
$A_t \ge K \log n/(1-\lambda)^2$ for some $t=\bigO(\log n /(1-\lambda)^3)$.
Lemma~\ref{lemma:fastphase2} shows that from this point additional 
$\bigO(\log(n)/(1-\lambda))$ rounds bring the infection size $A_t$ up to at least $(9/10)n$ \whp\
Lemma~\ref{lemma:expanderPhase3} shows that when $A_t$ becomes $\ge (9/10)n$, then 
\whp\ the whole graph becomes infected at some point within the subsequent 
$\bigO(\log(n)/(1-\lambda))$ rounds.


\begin{lemma}\label{lemma:fastphase2}
Let $G$ be a connected $n$-vertex $r$-regular graph with $\lambda<1$. Suppose that $|A_t| \geq K \log(n)/(1-\lambda)^2$, with $K = 4000$. Then the \BIPS process infects at least $9/10$ of the whole graph in $\bigO(\log(n)/(1-\lambda))$ extra rounds with probability at least $1-\bigO(1/n^3)$.
\end{lemma}

\begin{proof}
Assume $A_t$ has size less or equal than $9n/10$ but greater than $K \log(n)/(1-\lambda)^2$, then from Lemma \ref{Eval}
\begin{eqnarray*}
\E(A_{t+1}|A_t) & \geq & A_t(1+(1-\lambda^2)(1-9/10)) \\
& \geq & A_t\brac{1+\frac{1-\lambda}{10}}.
\end{eqnarray*}
Let $\varepsilon = \sqrt{10\log(n)/A_t}$. 
Observe that, given $A_t$, the size of $A_{t+1}$  is the sum of independent Bernoulli random variables. 
Using Chernoff bound for the lower tail of the sum of Bernoulli random variables, we get
\begin{eqnarray}
\Prob(A_{t+1} < (1-\varepsilon)\E(A_{t+1}|A_t)|A_t) 
  & \leq & e^{-\varepsilon^2 \E(A_{t+1}|A_t)/2} 
  \; = \;  e^{-5\log(n)} \; =\; \frac{1}{n^5}.\label{eqn:Hoeffding1}
\end{eqnarray}
By hypothesis $A_t \ge 4000\log n / (1-\lambda)^2$, so $\varepsilon \le (1-\l)/20$.
Therefore, with probability at least $1-n^{-5}$ we have
\begin{eqnarray*}
A_{t+1} & \geq & (1-\varepsilon)\E(A_{t+1}|A_t) \\
 & \geq & A_t\brac{1+\frac{1-\lambda}{10}} \brac{1-\frac{1-\lambda}{20}} \\
 & \ge & A_t \brac{1+\frac{1-\lambda}{23}}.
\end{eqnarray*}
Finally, we have that after $23/(1-\lambda)$ rounds, the size of infection has at least doubled. 
Hence, with probability at least $1-23\log(n)n^{-5}/(1-\lambda) \geq 1-n^{-4}$, 
after $23\log(n)/(1-\lambda)$ rounds, the infection covers at least $9n/10$ vertices.
\end{proof}

\begin{lemma}\label{lemma:expanderPhase3}
Let $G$ be a connected $n$-vertex $r$-regular graph with $1-\lambda \gg \sqrt{\log(n)/n}$.
With probability at least $1-n^{-5}$, after
 $T \le 8\log(n)/(1-\lambda)$ rounds the \BIPS process infects the whole graph.
\end{lemma}
\begin{proof}
For convenience, let $A_0$ and $B_0$ be the size of the infected and non-infected sets at the beginning of this phase and denote $q = 9/10$. Clearly $A_0 \ge qn$. Let  $A_t$ and $B_t$ be their respective sizes  after $t$ rounds.
From~\eqref{eqn:boundSizeAt|At-1} we get
\begin{eqnarray}
\E(A_{t+1}|A_t = A) &\geq&
A+(n-A)(1-\lambda^2)A/n. \label{eqn:ExpectedA_tbigSize}
\end{eqnarray}
The corresponding inequality for $B_{t+1}$ is
\begin{eqnarray}
\E(B_{t+1}|B_t) &\leq& B_t-B_t(1-\lambda^2)A_t/n  \nonumber \\
 & =& B_t (1-(1-\lambda^2)A_t/n).
\label{eqn:expectedB_i}
\end{eqnarray}
Let $|A_t|=k$. By applying the law of total probability and equation~\eqref{eqn:expectedB_i}, we get
\begin{eqnarray}
\E(B_{t+1})
&=& \!\! \sum_{k =qn}^n \E(B_{t+1}|B_t=n-k)\Prob(B_t =n-k) 
\: +\: \E(B_{t+1}|A_t< qn)\Prob(A_t < qn) \nonumber\\
&\leq & \!\!  \sum_{k = qn}^n (n-k)(1-(1-\l^2)k/n)\Prob(B_t = n-k) 
\:  +\: n\Prob(A_t < qn)\nonumber\\
&\leq& \!\!  \sum_{k = qn}^n (n-k)(1-(1-\l^2)q)\Prob(B_t = n-k) 
\: + \: n\Prob(A_t < qn)\nonumber\\
&\leq& \!\! (1-(1-\lambda^2)q)\E(B_{t})+n\Prob(A_t < qn).
\label{eqn:estimationE(B_i+1)}
\end{eqnarray}
We next prove that
\begin{eqnarray}\label{eqn:PA_t<qn}
 \Pr(A_{t} < qn)\leq  t n^{-8}.
\end{eqnarray}
To check the last inequality consider the event $E_t = \{A_t  \geq qn, i = 0,\ldots, t\}$. We are going to prove that $E_t$ has high probability. Indeed
\begin{eqnarray}
\Prob(E_t) &=& \Prob(E_t| E_{t-1})\Prob(E_{t-1}) +\Prob(E_t| E_{t-1}^c)\Prob(E_{t-1}^c) \nonumber\\
&\ge& \Prob(E_t| E_{t-1})\Prob(E_{t-1}).\nonumber
\end{eqnarray}
Observe that $A_t$ depends only on $A_{t-1}$ since it is a Markov chain, then 
$$\Prob(E_t| E_{t-1}) \; = \; \Prob(A_t \geq qn| A_{t-1} \geq qn),$$
and by a standard coupling argument
$$\Prob(A_t \geq qn| A_{t-1} \geq qn) \; \geq \;\Prob(A_t \geq qn| A_{t-1} = qn).$$
Choose $\varepsilon = \sqrt{16\log(n)/qn}$, then, by Chernoff bound
\begin{eqnarray*}
\Prob(A_{t+1} < (1-\varepsilon)\E(A_{t+1}|A_t=qn)|A_t=qn)
 & \leq & e^{-\varepsilon^2 \E(A_{t+1}|A_t=qn)/2} 
\; = \; e^{-8\log(n)}
\; =\; \frac{1}{n^8}. 
\end{eqnarray*}
Since we assume that $1-\lambda \gg \sqrt{\log(n)/n}$,
we have
$A_t = qn \geq 4000 \log(n)/(1-\lambda)^2$, 
so $\varepsilon \le (1-\l)/15$.
Thus with probability at least $1-n^{-8}$ we have
\begin{eqnarray*}
A_{t+1} & \geq & (1-\varepsilon)\E(A_{t+1}|A_t=qn) \\
 & \ge & qn\brac{1+\frac{1-\lambda}{10}}\brac{1-\frac{1-\lambda}{15}} \; \ge \; qn.
\end{eqnarray*}
Therefore $\Prob(A_t \geq qn| A_{t-1} = qn) \geq 1-n^{-8}$. We conclude that
 $$\Prob(E_t) \geq (1-n^{-8})^t \geq 1-tn^{-8}.$$
Observe that $\Prob(A_t \geq qn) \geq \Prob(E_t) \geq 1-tn^{-8}$, so~\eqref{eqn:PA_t<qn} holds.

Returning to our analysis, from Inequalities~\eqref{eqn:estimationE(B_i+1)} and \eqref{eqn:PA_t<qn} we have
\begin{equation}\label{vcxvjs78sw1}
  \E(B_{t+1}) \; \leq \; (1-(1-\lambda^2)q)\E(B_{t})+tn^{-7}.
\end{equation}
Denote $\theta = (1-(1-\lambda^2)q)$, then by iterating~\eqref{vcxvjs78sw1}
and using $B_0 = (1-q)n$, we get
$$\E(B_{t}) \; \leq \; \theta^t(1-q)n + \bigO(t^2n^{-7}) \leq n\theta^t+\bigO(t^2n^{-7}).$$
Choosing $T = 5\log(n)/\log(1/\theta)$ and applying Markov's inequality give
\begin{eqnarray}
\Prob(B_T \geq 1) & \leq & \E(B_T) \; \leq \; \theta^T+\bigO(T^2n^{-7}) \nonumber \\
& = & n^{-5}+\bigO(T^2n^{-7}). \label{nkkio46a}
\end{eqnarray}
Finally, observe that for $0 < \th <1$ we have $(1-\th) \le \log (1/\th)$ and thus
\begin{eqnarray}
T &=& 6 \log(n)/(\log(1/\theta)) \; \leq \; 6 \log(n)/(1-\theta) \nonumber \\
 & \leq & 6 \log(n)/(q(1-\lambda^2)) 
 \; \leq \;  6\log(n)/(q(1-\lambda)) \nonumber \\
& \leq & 8\log(n)/(1-\lambda)
\; = \; \bigO(n), \label{nbmzvw}
\end{eqnarray}
where the last bound follows from the assumption $1-\lambda \gg \sqrt{(\log n)/n}$.
We obtain $\Prob(B_T \geq 1) = \bigO(n^{-5})$
from~\eqref{nkkio46a} and~\eqref{nbmzvw}.
\end{proof}

\begin{sloppy}

\begin{proof}[of Theorem~\ref{Th2}]
Apply Lemma~\ref{lemma:fastphase1} with $m = 4000(\log n)/(1-\lambda^2)$. 
Observe that since we assume $1-\lambda \gg \sqrt{(\log n)/n}$, 
there is no problem with the restriction $m \leq n/2$ in this lemma. 
After that a straightforward application of Lemmas~\ref{lemma:fastphase2} 
and~\ref{lemma:expanderPhase3} gives us the result.
\end{proof}

\end{sloppy}

\newpage

\bibliographystyle{abbrv}

\end{document}